\documentclass[11pt]{article}
\usepackage[margin=1in, paperwidth=8.5in, paperheight=11in]{geometry}
\usepackage{graphicx}
\usepackage{float}
\usepackage{caption}
\usepackage{longtable}
\usepackage{booktabs} 

\usepackage{color}
\usepackage{amsmath}
\usepackage{cancel}
\usepackage{tikz}
\usepackage{amsfonts}
\usepackage[linesnumbered,ruled]{algorithm2e}

\usepackage{dsfont}
\usepackage[titletoc,title]{appendix}

\usepackage[T1]{fontenc}
\usepackage[utf8x]{inputenc} 
\usepackage[affil-it]{authblk}
\usepackage{lmodern}
\usepackage[english]{babel}
\usepackage{amsfonts}
\usepackage{amssymb}
\usepackage{ stmaryrd }
\usepackage{ textcomp }
\usepackage{ bbold }
\usepackage{authblk}
\usepackage{mathrsfs }
\usepackage[all]{xy} 
\usepackage{xcolor}
\usepackage{tikz} 
\usepackage{babel}
\usepackage{amsthm}
\usepackage{verbatim}
\usepackage{appendix}

\usepackage{setspace}
\usepackage{listings}
\PrerenderUnicode{é}

\newtheorem{theorem}{Theorem}[section]
\newtheorem{lemma}[theorem]{Lemma}
\newtheorem{proposition}[theorem]{Proposition}
\newtheorem{corollary}[theorem]{Corollary}
\newtheorem{assumption}[theorem]{Assumption}

\newenvironment{definition}[1][Definition]{\begin{trivlist}
\item[\hskip \labelsep {\bfseries #1}]}{\end{trivlist}}

\providecommand{\keywords}[1]{\textbf{\textit{Keywords---}} #1}
\def\mathbi#1{\textbf{\em #1}}

\begin{document}

\title{Hamiltonian Flow Simulation of Rare Events}
\author[1]{Raphael Douady}
\author[2]{Shohruh Miryusupov  \thanks{Corresponding author: \texttt{shohruh.miryusupov@malix.univ-paris1.fr}}}
\affil[1]{Stony Brook University, CNRS, Université Paris 1 Panthéon Sorbonne}
\affil[2]{Université Paris 1 Panthéon Sorbonne\\ Labex RéFi \footnote{This work was achieved through the Laboratory of Excellence on Financial Regulation (Labex RéFi) supported by PRES heSam under the reference ANR­10­LABX­0095. It benefited from a French government support managed by the National Research Agency (ANR) within the project Investissements d'Avenir Paris Nouveaux Mondes (Investments for the Future Paris ­New Worlds) under the reference ANR­11­IDEX­0006­02.}}
\maketitle
\begin{abstract}
Hamiltonian Flow Monte Carlo(HFMC) methods have been implemented in engineering, biology and chemistry. HFMC makes large gradient based steps to rapidly explore the state space. The application of the Hamiltonian dynamics allows to estimate rare events and sample from target distributions defined as the change of measures. The estimates demonstrated a variance reduction of the presented algorithm and its efficiency with respect to a standard Monte Carlo and interacting particle based system(IPS). We tested the algorithm on the case of the barrier option pricing.
\end{abstract}
\tableofcontents

\section{Introduction}

\hspace{3ex} Hamiltonian flow based Monte Carlo simulations originates from physics have been used in many applications in statistics, engineering for a number of years. However these methods are not widely used in the estimation of rare events and in the financial option pricing practise.

This paper proposes Hamiltonian Flow Monte Carlo technique for an efficient estimation of the rare event probability. Similarly to an importance sampling technique this method involves a change of probability measure. The random variables are sampled according to a modified probability measure that differs from a reference measure.

A rare event is the probability $\mathbb{P}(f(x_t)>a_t)$ for large values of $x$. One way to deal with this problem is to change an original measure, so that $a_k$ is not too large in the new measure.
Define the following set:
\begin{equation}
A_t=\{x\in \mathbb{R}^M, f(x_t)>a_t\}
\end{equation}
\begin{equation}
\mathbb{E}^{\mathbb{P}}[\mathbb{1}_{A_t}]=\mathbb{E}^{\mathbb{Q}}[\mathbb{1}_{A_t}L_t]
\end{equation} 
where $L=\frac{d\mathbb{P}}{d\mathbb{Q}}$ is a Radon-Nykodim derivative.

The Hamiltonian approach in the Monte Carlo context was developed by Duane et al. \cite{Duane1987}, R. Neal \cite{N2010} where they proposed an algorithm for sampling probability distributions with continuous state spaces. The advantage of the Hamiltonian based Monte Carlo is in the fact that we could extend the state space by including a momentum variable that will force in our context to move long distances in the state space in a single update step. We use this property of the Hamiltonian dynamics to explore rarely-visited areas of the state space and efficiently estimate rare-event probability. Algorithm consists of two parts: simulation of the Hamiltonian dynamics and Metropolis-Hastings test, that removes the bias and allows large jumps in the state space. We will show the performance of the algorithm on the Down-Out Barrier option technique with low level of a barrier.

HFMC could be considered within the optimal transportation problem, which was posed back in the 18-th century. Like HFMC, other simulation based approaches such as particle methods \cite{DM2013},\cite{DMG2005} or the transportation using a homotopy \cite{RDSM2017} allow to move a set of particles from the measure $\mathbb{P}$ to the measure $\mathbb{Q}$, by minimizing the transportation cost. We will show how the rare events estimation could be computed using interacting particle systems \cite{DM2004}.

The paper is organized as follows. Section 2 introduces state of the art approach to estimate rare events: MC and PMC and formulates the problem within the context of Barrier option pricing. Section 3. describes the Hamiltonian Flow Monte Carlo Algorithm. Section 4 presents the law of large numbers and convergence for HFMC. An adaptation of the Hamiltonian Flow Algorithm in the case of a Barrier option, numerical results and discussion is presented in Section 5. Section 6 concludes.

\section{Monte Carlo and Interacting Particle System}

\subsection{Problem Formulation}
\hspace{3ex} Barrier option pricing is one of the cases when we encounter to the case of rare events. Consider a sequence of random variables $\{X_n\}_{n\geq 0}$, which in the financial context could be interpreted as asset prices, which forms a Markov Chain on the space $\mathbb{R}^{n_S}$. Given some stochastic process $\{X_t\}_{t\geq 0}$, for any test function $f$, we would like to compute the following expectation:

\begin{equation} \label{expec}
C=\mathbb{E}^{\mathbb{P}}[f(X_T)\mathbb{1}_{X_{t \in [0,T]} \in A_t}]
\end{equation}

One of the most popular ways to deal with this problem is importance sampling, when we replace the original statistical measure $\mathbb{P}$ by an importance measure $\mathbb{P}^{\delta}$. Then for $\textbf{X}_n=(X_1,...,X_n)$.:

\begin{equation} 
C=\int f(X_{n_t}) \frac{d\mathbb{P}}{d\mathbb{P}^{\delta}}(\textbf{X}_{n_t})d\mathbb{P}^{\delta}(\textbf{X}_{n_t})\prod_{n=1}^{n_t} \mathbb{1}_{X_{n} \in A_n} = \mathbb{E}^{\mathbb{P}^{\delta}}\left[f(X_{n_t})\frac{d\mathbb{P}}{d\mathbb{P}^{\delta}} \prod_{n=1}^{n_t}\mathbb{1}_{X_{n} \in A_n}\right]
\end{equation}
where the rare event set $A_n$ is given by:
\begin{equation}
A_n=\{X_n\in \mathbb{R}^{n_S}, f(X_n)>a_n\}
\end{equation}
In practise we don't have an explicit form of the likelihood ratio $\frac{d\mathbb{P}}{d\mathbb{P}^{\delta}}$, so it becomes unfeasable unless one considers very simple toy examples. One of the solutions was interacting particle system(IPS), which was proposed by Del Moral and Garnier \cite{DMG2005}, where they proposed to generate particles(samples) in two steps, i.e. particle mutation and selection. The idea is to approximate the ratio of $\mathbb{P}$ with respect to some importance measure $\mathbb{P}^{\delta}$ by choosing a weight function $\omega_n$ that approximates the Radon-Nikodym derivative $\frac{d\mathbb{P}}{d\mathbb{P}^{\delta}}$. If we assume that $\mathbb{P}$ and $\mathbb{P}^{\delta}$ have density function $p$ and $\widetilde{p}$ respectively, for $n_S$ particles $\{X_n^{(m)}\}_{m=1}^{n_S}$ we could define the weight function iteratively by:
\begin{equation}
\prod_{n=1}^{n_t}\omega_n (\textbf{X}_n^{(m)}) \propto \frac{d\mathbb{P}}{d\mathbb{P}^{\delta}}(\textbf{X}_{n_t}^{(m)})=\prod_{n=1}^{n_t}\frac{p_n(X^{(m)}_n,X^{(m)}_{n-1})}{p_n^{\delta}(X^{(m)}_n,X^{(m)}_{n-1})}
\end{equation}

Since two measures $\mathbb{P}$ and $\mathbb{P}^{\delta}$ form a Markov chain, the Radon-Nykodim derivative could be decomposed into the product of ratio of the transition density $p_n(\cdot, X^{(m)}_{n-1})$ to the transition density with respect to the measure $\mathbb{P}^{\delta}$.

The normalized importance weight function is given by:

\begin{equation}
W_n(\textbf{X}_{n}^{(m)})=\frac{\omega_n(\textbf{X}_{n}^{(m)})}{\frac{1}{n_S}\sum_{j=1}^{n_S} \omega_n(\textbf{X}_{n}^{(j)})}
\end{equation}

The IPS estimate of an expectation (\ref{expec}) will have the following form:
\begin{equation}
\widehat{C}^{IPS}=\frac{\mathbb{E}\left[f(X_{n_t})\prod_{n=1}^{n_t-1} \omega_{n}(X_{n}) \mathbb{1}_{X_n \in  A_n} \right]}{\mathbb{E}\left[\prod_{n=1}^{n_t-1} \omega_{n}(X_{n}) \mathbb{1}_{X_n \in  A_n} \right]}
\end{equation}



In our experiments we choose a potential function(an unnormalized importance weight) of the form:
\begin{equation}
\omega_n(\textbf{X}_{n_t}^{(m)})=\prod_{n=1}^{n_t} e^{\delta(X_n^{(m)}-X_{n-1}^{(m)})}
\end{equation}

where $\delta$ is an exponential tilting parameter. One issue with this approach is in the fact that an optimal choice of tilting parameter $\delta$ has to be judiciously chosen by running simulations and, in fact, it is fixed across all time steps $n=1,...,n_t$.

\subsection{Algorithm}
\hspace{3ex}  The  algorithm could be described by the following scheme, $\forall m = 1,...,n_S$:
\begin{equation}
X_{n}^{(m)}  \xrightarrow[]{Sampling} \widehat{X}_{n+1}^{(m)} \xrightarrow{Selection} \Phi(\widehat{X}_{n+1}^{(m)})= X_{n+1}^{(m)}
\end{equation}

At each time step $n=1,...,n_t$ we draw $n_S$ independent random variables from the density $p_n^{\delta}(\cdot,X_{n-1}^{(m)})$ to construct $n_S$ particles, $\widehat{\textbf{X}}_{n}^{(m)}=(\widehat{X}_{n}^{(m)},\widehat{\textbf{X}}_{n-1}^{(m)})$. Given generated particles, we select, or draw independently $n_S$ particles $\textbf{X}_{n}^{(m)}=(X_{0}^{(m)},...,X_{n}^{(m)})$ with replacement of rejected particles according to their probability weights:

\begin{equation}
W_n(\widehat{\textbf{X}}_{n}^{(m)})=\frac{\omega_n(\widehat{\textbf{X}}_{n}^{(m)})}{\frac{1}{n_S}\sum_{j=1}^{n_S} \omega_n(\widehat{\textbf{X}}^{(j)}_n)}
\end{equation}

And at time step $n_t$, we get the following IPS estimator: 

\begin{equation*}
\mathbb{E}[f(X_{n_t})\prod_{n=1}^{n_t} \mathbb{1}_{X_n \in A_n}]\approx\frac{1}{n_S}\sum_{m=1}^{n_S} \left(f(\widehat{X}_{n_t}^{(m)})\prod_{n=1}^{n_t} \omega_{n-1}(\textbf{X}_{n-1}^{(m)}) \mathbb{1}_{\widehat{\textbf{X}}_{n}^{(m)} \in A_n} \right)
\end{equation*}



\begin{algorithm}[H] 
Initialization: $n_S$ - $\#$(simulations), $n_t$ - $\#$(time steps), $X_0$ - initial value\\

\For{$n=1,...,n_t$}{
\For{$m=1,...n_S$}{
Generate $X_{n}^{(m)}$ from $p(\cdot,X_{n-1}^{(m)})$ and set $\widehat{\textbf{X}}_{n}^{(m)}=(\widehat{X}_{n}^{(m)},\textbf{X}_{n-1}^{(m)})$;\\
\If{$\widehat{\textbf{X}}_{n}^{(m)} \in A$}{
$\widehat{X}_{n}^{(m)}=0$
\\
\Else{
Compute the weight: $\omega_n(\widehat{\textbf{X}}_{n}^{(m)})$.
}
}
}

\If{$n<N$}{

Resample using probability weight: $W_n(\widehat{\textbf{X}}_{n}^{(m)})=\frac{\omega_n(\widehat{\textbf{X}}_{n}^{(m)})}{\frac{1}{n_S}\sum_{j=1}^{n_S} \omega_n(\widehat{\textbf{X}}^{(j)}_n)}$ to sample $\textbf{X}_n^{(m)}$.

}
}


    \caption{IPS algorithm}
\end{algorithm}

\section{Hamiltonian Flow Monte Carlo}

\subsection{Markov Chains on a Phase Space}
\hspace{3ex} From section 2 we know that one of the ways to deal with rare event probabilities is to change a measure:
\begin{equation} \label{RN}
\mathbb{E}^{\mathbb{P}}[f(X_{n_t})\prod_{n=1}^{n_t} \mathbb{1}_{X_n \in A_n}]=\mathbb{E}^{\mathbb{Q}}[f(X_{n_t})\frac{d\mathbb{P}}{d\mathbb{Q}}\prod_{n=1}^{n_t} \mathbb{1}_{X_n \in A_n}]
\end{equation}

To approximate Radon-Nikodym derivative $\frac{d\mathbb{P}}{d\mathbb{Q}}$ we will generate Markov Chain that will converge to an ergodic distribtion. Let us introduce a random process $X_u$ in a pseudo-time $u$ and consider the following SDE, which is a gradient flow distrubed by a noise:
\begin{equation} \label{Sm}
dX_u=-\nabla \Psi(X_u)du+2\sqrt{\beta^{-1}}dW_u
\end{equation}

where $\Psi(X):=-\log(p(X))$ is a potential. Under the assumption of ergodicity, the auto-correlated path $X_u$ asymptotically, i.e. $u\rightarrow \infty$ draws samples from a stationary distribution:

\begin{equation}
\pi(X)=\frac{1}{\mathcal{Z}}\exp(\Psi(X))
\end{equation}
where a normalizing constant $\mathcal{Z}$:
\begin{equation}
Z=\int_{\mathbb{R}^n} e^{-\beta\Psi(x)} dx
\end{equation}

This can be seen as a unique solution of the following Fokker-Plank equation, given that $\Psi$ satisfies to some growth condition:
\begin{equation} \label{FP}
\frac{\partial p(t,x)}{\partial t} =div (\nabla(\Psi(x)\rho)) +\beta^{-1} \Delta p
\end{equation}

When we mentioned ergodicity, we meant, that for a class of regular functions $\phi:\mathbb{R}^X\rightarrow \mathbb{R}$ and $x_0$ a.s., the Markov Chain satisfies:
\begin{equation}
\frac{1}{L}\sum_{l=1}^L\phi(x_l)\rightarrow \int_{\mathbb{R^X}}\phi(x)\pi(dx)=\mathbb{E}^{\pi}[\phi(X)]
\end{equation}

Observe that eq. (\ref{Sm}) is a reversible process, which is interesting from theoretical point of view, but in practise the speed of convergence is not optimal.  One of the ways to improve the convergence is to add a divergence-free drift $b$ and consider the following modified SDE:
\begin{equation}
dX_u=(-\nabla \Psi(X_u)+b(X_t))du+2\sqrt{\beta^{-1}}dW_u
\end{equation}
in order to satisfy detailed balance condition, we assume that $\nabla(be^{-\Phi})=0$.

Another way to improve the convergence is to consider a generalized Langevin SDE:
\begin{equation}
 \ddot{X}_u^{\gamma}=-\nabla \Psi(X_u^{\gamma})-\dot{X}_u^{\gamma}+\sqrt{2\beta^{-1}}\dot{W}_u
\end{equation}

We can rewrite it as:
\begin{equation} 
\left\{\begin{array}{c}

dX_u = P_udu \\
dP_u = -\nabla \Psi(X_u)du-P_udu+\sqrt{\frac{2}{\beta}}dW_u 

\end{array}\right.
\end{equation}

where the pair $(X,P)$ is a kinetic process with $X$ is the position and $P=\frac{dX}{du}$ is the velocity, that acts as an instantaneous memory.

The invariant function of the Markov process $\{x,P\}$, if it exists, is given by:
\begin{equation}
\pi_0(x,P)=\frac{1}{\mathcal{Z}}e^{-\beta \mathcal{H}(x,P)}, \ \ \mathcal{Z}=\int_{\mathbb{R}^2} e^{-\beta \mathcal{H}(x,P)}dPdx
\end{equation}
where 
\begin{equation} \label{ham}
\mathcal{H}(x,P)=\frac{1}{2}P\mathfrak{M}^{-1}P+\Psi(x)
\end{equation}
is a Hamiltonian function on $\mathbb{R}^2$. 

We will use the Hamiltonian system to generate The Markov Chain and approximate a Radon-Nikodym derivative $\frac{d\mathbb{P}}{d\mathbb{Q}}$. The Hamiltonian Flow Monte Carlo uses a physical simulation of moving particles with momentum under the impact of the energy function to propose Markov chain transitions, that allows rapidly explore state space. Its fast exploration can be explained by the fact that it extends the state space by an auxilliary momentum variables, $P$, and then runs a physical simulation to move long distances along probability contours in the extended state space.  

We remind that, given the Markov Chain $\{X_l\}_{l\geq 0}$, Birkhoff theorem says that
\begin{equation}
\frac{1}{n_S}\sum_{l=1}^{n_S} f(X_l)\underset{n_S \rightarrow \infty}{\longrightarrow} \int f(x)d\pi(dx)=\varrho \ \mbox{a.s.}
\end{equation}
where $\varrho$ is the expectation of $f(X)$ with respect to the unique invariant distribution $\pi$ of the Markov Chain.

\subsection{Hamiltonian Flow's Integrator and Properties}
We will use a configuration space $\mathbi{M}$ with periodic boundary conditions. Each point on $\mathbi{M}$ will be a set of $n_S$ particles: $X^{(1)},...,X^{(n_S)}$ and a generic momentum space $\mathbb{R}^{n_S}$, in this case the cotangent space is given by $T^{*}\mathbi{M}=\mathbb{R}^{n_S}\times \mathbb{R}^{n_S}$.


\begin{equation}
\begin{aligned}
\Xi_u: T^{*}\mathbi{M} \rightarrow T^{*}\mathbi{M}
\\
(X,P)\rightarrow \Xi_u(X,P)
\end{aligned}
\end{equation}

$\Xi_u(X_0,P_0)$ is the solution to the Hamilton's equation: 

\begin{equation} \label{HS}
\left\{\begin{array}{c}

dX_u = \mathfrak{M}^{-1} P_udu \\
dP_u = -\nabla \Psi(X_u)du

\end{array}\right.
\end{equation}

Hamiltonian system has three main properties: reversibility, conservation of energy and volume preservation.

\subsubsection{Sympletic Integration of Hamiltonian Equations}

In most cases we can not compute the Hamiltonian flow in closed form and that is why we need to discretize the system (\ref{HS}). To make sure that we can preserve symplecticness and reversibility, we will discretize using leap-frog integrator, which is a symplectic integrator of the Hamiltonian system.

Split the Hamiltonian (\ref{ham}) into 3 parts:
\begin{equation}
\mathcal{H}_1=\frac{1}{2}\Psi(X), \ \ \mathcal{H}_2=\frac{1}{2}\langle P,\mathfrak{M}^{-1} P \rangle, \ \  \mathcal{H}_3=\frac{1}{2}\Psi(X)
\end{equation}
Taking each of these terms separately to be the Hamiltonian function of a Hamiltonian
system gives rise to equations of motion with trivial dynamics. 
\begin{equation}
\left\{ \begin{array}{c} P_n(u+\frac{\Delta u}{2})=P_n(u)-\frac{\Delta u}{2}\frac{\partial \Psi}{\partial x_n} (X_n(u)) \\ 
X_n (u+\Delta u)=X_n(u)+\Delta u P_n(u+\frac{\Delta u}{2})\mathfrak{M}^{-1} \\ 
P_n(u+\frac{\Delta u}{2})=P_n(u+\frac{\Delta u}{2})-\frac{\Delta u}{2}\frac{\partial \Psi}{\partial x_n}(X_n(u+\Delta u))
\end{array} \right.
\end{equation}
where $\Delta u$ is the discretization size of the Hamiltonian.

Consider a concatenation of three maps:
\begin{equation}
\Xi_{n}=\Xi_{\Delta u, \mathcal{H}_3}\circ \Xi_{\Delta u, \mathcal{H}_2}\circ \Xi_{\Delta u, \mathcal{H}_1}
\end{equation}
where $\Xi_{\Delta u, \mathcal{H}_1}:(X(0),P(0))\rightarrow (X(\Delta u),P(\Delta u))$. Similarily, $\Xi_{\Delta u , \mathcal{H}_1}=\Xi_{\Delta u , \mathcal{H}_3}$, and $\Xi_{\Delta u , \mathcal{H}_2}$ is calculated to be position update.
Since the energy is preserved by the flow, the trajectories evolve on the submanifold of constant energy:
\begin{equation}
T^{*}\mathbi{M}(E_0)=\{(X,P)\in T^{*}\mathbi{M};(\mathcal{H}(X,P)=E_0)\}
\end{equation}
where $E_0=\mathcal{H}(X_0,P_0)$ is the energy of the initilized data.
\subsection{Hamiltonian Flow Monte Carlo on Rare Events Sets}

Let $\mathcal{H}(X,P)$ be the Hamiltonian function on $\mathbb{R}^{2n_S}$, where $X$ is a potential, and $P$ is a momentum variable of the Hamiltonian system. The algorithm consists of two steps, first sampling from prior distribution values for potential and momentum and then a physical simulation of the Hamiltonian dynamics. To make sure that at the end of each physical simulation of time step $n+1$ we will have a probability measure, i.e. values will not exceed $1$, we will use a Metropolis-Hastings test $\alpha_{n+1}$, by choosing the minimum between $1$ and the ratio of generated values of potential at time steps $n+1$ and $n$, which is an acceptance probability of potential simulated by the Hamiltonian dynamics. If we extend the state space $X=\{X_1,...,X_n\}$ and denote the extended space as $\widetilde{X}=\{X_1,...,X_n,P_1,...,P_n \}$, we can denote the acceptance probability as:

\begin{equation}
\alpha_{n+1}(\widetilde{X}_{n},\widetilde{X}_{n+1})=1 \wedge e^{(-\mathcal{H}(X_{n+1},P_{n+1})+\mathcal{H}(X_{n},P_{n}))\Delta t}
\end{equation}

If we assume that the importance measure $\mathbb{Q}$ admits the following importance distribution with a kernel $\mathcal{K}$:
\begin{equation}
q(d\widetilde{X}_{n+1})=\int_{\mathbb{R}^{2M}} p(\widetilde{X}_n)\mathcal{K}(\widetilde{X}_n,d\widetilde{X}_{n+1})d\widetilde{X}_n
\end{equation}

Then, the associated Radon-Nikodym derivative will have the following form:
\begin{equation}
\frac{d\mathbb{P}}{d\mathbb{Q}}(\widetilde{X}_{n+1})=\frac{d\mathbb{P}(\widetilde{X}_{n+1})}{\int_{\mathbb{R}^{2M}} p(\widetilde{X}_n)\mathcal{K}(\widetilde{X}_n,d\widetilde{X}_{n+1})d\widetilde{X}_n}
\end{equation}

Assume that at each times step $n$ we have $n_S$ sample of r.v. $\{X_n^{(m)}\}_{m=1}^{n_S}$. Now we can define a transition kernel $\mathcal{K}$ as follows.
\begin{definition}
Consider a mapping $\Xi_n:\widetilde{X}_n^{(m)} \rightarrow \widetilde{X}_{n+1}^{(m)}$, which is a transformation in $\mathbb{R}^{2n_S}$, $\mathfrak{u} \sim Unif[0,1]$.  Then a transition kernel $\mathcal{K}(\cdot,d\widetilde{X}^{(m)}_{n+1})$ is given by:
\begin{equation}
\mathcal{K}(\widetilde{X}^{(m)}_{n},d\widetilde{X}^{(m)}_{n+1})=\mathbb{1}_{\mathfrak{u}\leq \alpha_{n+1}}\Xi_n(\widetilde{X}_n^{(m)})d\widetilde{X}^{(m)}_{n+1}+\mathbb{1}_{\mathfrak{u} > \alpha_{n+1}}\widetilde{X}^{(m)}_{n}\delta_{\widetilde{X}^{(m)}_{n}}(d\widetilde{X}^{(m)}_{n+1})
\end{equation}
\end{definition}
This kernel can be interpreted as the probability to move from the point $\widetilde{X}_n^{(m)}$ to a new proposed point $\widetilde{X}_{n+1}^{(m)}$, which is simulated through a discretized Hamiltonian flow $\Xi_n(\cdot)$ . If the proposed step is not accepted, then next step is the same as the current step, i.e. $\widetilde{X}_{n+1}^{(m)}=\widetilde{X}_{n}^{(m)}$. This procedure allows as to leave the joint distribution of $X_n^{(m)}$ and $P_n^{(m)}$ invariant. Volume preservation means that the determinant of the Jacobian matrix of a transformation $\Xi_n$ is equal to one.

We will need basic property of symplectic integrators, i.e. reversibility.
\begin{lemma}
The integrator $\Xi_n$ is reversible.
\end{lemma}
We refer to \cite{HS2005} for the proof of this result.
\begin{assumption} \label{ass1}
\begin{itemize}
\item The potential $\Psi\in \mathcal{C}^{1}$ is bounded from above;
\item The gradient $\nabla\Psi$ is a globally Lipschitz function.
\end{itemize}
\end{assumption}

\begin{lemma}
If the potential $\Psi$ satisfies to the assumption \ref{ass1}, then the kernel $\mathcal{K}$ is irreducible and the Markov Chain satisfies
\begin{equation}
\forall x \in \mathbf{M}, \forall B\in \mathcal{B}(\mathbf{M}), \mu^{Leb}(B)>0, \mathcal{K}(x,B)>0
\end{equation}
\end{lemma}
\begin{proof}
We refer to $\cite{CLS}$.
\end{proof}
\begin{proposition}
Given that the assumption \ref{ass1} holds, then for $n=1,...,n_S$, the irreducible Markov Chain defined by a transformation $\Xi_n$ is reversable under the distribution $\pi$:
\begin{equation}
\pi(d\widetilde{X}^{(m)}_{n})\mathcal{K}(\widetilde{X}^{(m)}_{n},d\widetilde{X}^{(m)}_{n+1})=\pi(d\widetilde{X}^{(m)}_{n+1})\mathcal{K}(\widetilde{X}^{(m)}_{n+1},d\widetilde{X}^{(m)}_{n})
\end{equation}
Thus $\pi(x)$ is the invariant distribution of the Markov Chain $\{\widetilde{X}_{n}\}_{n=1}^{n_S}$.
\end{proposition}

\begin{proof}
Rewrite the kernel $\mathcal{K}$ as:
\begin{equation}
\mathcal{K}(x,dy)=\alpha(x,y)\Xi_n(x)dy+\mathbb{a}(x) \delta_{x}(dy)
\end{equation}
where
\begin{equation}
\mathbb{a}(x)=1-\int \alpha(x,z)\Xi_n(x)dz
\end{equation}

\begin{multline}
\int \mathcal{K}(x,B)\pi(x)dx= \int \left[\int_B \alpha(x,y)\Xi_n(x)dy \right]\pi(x)dx+\int \mathbb{a}(x)\delta_x(B)\pi(x)dx=\\=\int_B \left[\int \pi(x)\alpha(x,y)\Xi_n(x)dx \right]dy+\int_B \mathbb{a}(x)\pi(x)dx=\\=\int_B \left[\int \pi(y)\alpha(y,x)\Xi_n(y)dx \right]dy+\int_B \mathbb{a}(x)\pi(x)dx
\\=\int_B \pi(y)(1-\mathbb{a}(y))dy+\int_B \mathbb{a}(x)\pi(x)dx=\int_B \pi(y)dy
\end{multline}
\end{proof}


\begin{corollary}
The kernel $\mathcal{K}$ satisfies reversibility condition with an indicator function of the rare event set:
\begin{equation} \label{cor1}
\pi(x)\mathcal{K}(x,y)\mathbb{1}_{x\in A}=\pi(y)\mathcal{K}(y,x)\mathbb{1}_{y\in A}
\end{equation}
\end{corollary}
Now we can define rare event transitions  through the kernel $\mathcal{M}$. 
\begin{definition}
Assume that the assumption \ref{ass1} holds and consider a Markov Chain $(X_n^{(m)})_{n\geq 1}$ with an initial prior $p_1(X_1)$ and define the following transition kernel $p(\widetilde{X}^{(m)}_{n+1}\in d\widetilde{X}^{(m)}_{n+1}|\widetilde{X}^{(m)}_{n})=\mathcal{M}(\widetilde{X}^{(m)}_{n},d\widetilde{X}^{(m)}_{n+1})$.
\begin{equation}
\mathcal{M}(\widetilde{X}^{(m)}_{n},d\widetilde{X}^{(m)}_{n+1})=\mathcal{K}(\widetilde{X}^{(m)}_{n},d\widetilde{X}^{(m)}_{n+1})\mathbb{1}_{\mathcal{K}(\widetilde{X}^{(m)}_{n},d\widetilde{X}^{(m)}_{n+1}) \in A_n}+\widetilde{X}^{(m)}_{n}\mathbb{1}_{\mathcal{K}(\widetilde{X}^{(m)}_{n},d\widetilde{X}^{(m)}_{n+1}) \not\in A_n}
\end{equation}
\end{definition}

It means that the point $\widetilde{X}^{(m)}_{n}$ moves to a new point $\widetilde{X}^{(m)}_{n+1}$ only if it is inside a rare event set $A_n$, otherwise we stay at point $\widetilde{X}^{(m)}_{n}$.

\begin{proposition}
\hspace{1.5ex} Let $n=1,...,n_t$. The Markov chain $X_n$ is invariant under the kernel $\mathcal{M}(\cdot,dX_{n+1})$.
\end{proposition}
\begin{proof}
\begin{multline}
\int \pi(dx)\mathcal{M}(x,dy)\mathbb{1}_{x\in A}=
\int \pi(dx)\left[ K(x,y)\mathbb{1}_{x\in A}+K(x,A^c)\delta_x(dy)\right]\mathbb{1}_{x\in A}=\\=\int \int \pi(dx)K(x,dz)\left[ \mathbb{1}_{z\in A}\delta_z(dy)+\mathbb{1}_{A^c}(z)\delta_x(dy)\right]\mathbb{1}_{x\in A}=\\=\int  \pi(dx)K(x,dy) \mathbb{1}_{y\in A}\mathbb{1}_{x\in A}+\int \pi(dy)K(y,dz)\mathbb{1}_{A^c}(y)\mathbb{1}_{x\in A}=\pi(dy)\mathbb{1}_{y\in A}
\end{multline}
\end{proof}

Invariance of $\widetilde{X}_n^{(m)}$ says that for any bounded and measurable function $f$, the distribution of $f(\mathcal{M}(\widetilde{X}^{(m)}_{n},d\widetilde{X}^{(m)}_{n+1}))$ and $f(\widetilde{X}^{(m)}_{n})$ is the same.

\begin{equation}
\mathbb{E}[f(\mathcal{M}(\widetilde{X}^{(m)}_{n},d\widetilde{X}^{(m)}_{n+1}))]=\mathbb{E}[f(\widetilde{X}^{(m)}_{n})]
\end{equation}


Under the kernel $\mathcal{M}$ of $\widetilde{X}_n$, the final HFMC estimate is given by:
\begin{equation} \label{HFMC}
\widehat{C}^{HFMC}=\frac{1}{n_S}\sum_{m=1}^{n_S} f(X_{n_t}^{(m)})  \mathbb{1}_{\{X_{n+1}^{(m)}, X_{n}^{(m)} \in A_n\}}
\end{equation} 

\subsection{Algorithm}
\hspace{3ex} The Hamiltonian function is defined by $\mathcal{H}(X,P)=\Psi(X)+\frac{1}{2}P^T \mathfrak{M}^{-1}P$, where $\Psi(X)$ - is a potential energy function, and the second term is a kinetic energy function with a momentum variable $P$ and mass matrix $\mathfrak{M}$. Usually one sets a mass matrix $\mathfrak{M}$ to be an identity matrix $I$. The proposed samples are obtained by a physical simulation of the Hamiltonian dynamics:
\begin{equation}
\left\{ \begin{array}{c}
dX_u = \mathfrak{M}^{-1}P_udu \\
dP_u = -\nabla \Psi(X_u)du
\end{array} \right.
\end{equation}

We start by simulating $M$ random variables from  a prior $X_1=p_0(\cdot,X_{0}^{(m)})$, which is the density of the underlying SDE and generating $M$ random variables from gaussian distribution for momentum $\{P_0^{(m)}\}_{m=1}^M$.

For each step $n=1,...,N$ we set $x_H^{(m)}=X_n^{(m)}$, $P_H^{(m)}=P_n^{(m)}$. The proposed new candidates are obtained after $L$-leapfrog steps of the simulation of the Hamiltonian dynamics and they are defined by $x^*=x^{(m)}_H(L)$ and $P^*=P^{(m)}_H(L)$. These new set of proposed candidates are then accepted according to the following Metropolis-Hastings test. First generate uniformly distributed random variable $\mathfrak{u}\sim \mathcal{U}nif(0,1)$, then compute $\alpha$:

\begin{equation}
\alpha = 1 \wedge e^{(-\mathcal{H}(x^*,P^*)+\mathcal{H}(x_{H}^{(m)},P_{H}^{(m)}))\Delta t};
\end{equation}

If proposed candidates $(x^*,P^*)$ are accepted, i.e. $\alpha>\mathfrak{u}$ we set $X^{(m)}_{n+1}=x^*$, and if they are rejected, i.e. $\alpha \leq u$, we set $X^{(m)}_{n+1}=x^{(m)}_H$. At the end, calculate estimator in (\ref{HFMC}). The main steps of the algorithm are summarized in Algorithm 2.

The Metropolis-Hastings test insures a volume preservation. That explains the fact that we don't need to compute a normalizing constant in our algorithm. Volume preservation means that the absolute value of the Jacobian matrix of the leapfrog integrator is equal to one, this is because candidates are proposed though simulation of the Hamiltonian flow.

\begin{algorithm}[H] 
Initialization: $n_S$ - $\#$(simulations), $n_t$ - $\#$(time steps) \\

\For{$n=1,...,n_t$}{
\For{$m=1,...n_S$}{
Generate $X_n^{(m)}$ from prior $\widetilde{p}(X_0^{(m)},\cdot)$;\\
Simulate initial momentum $P^{(m)}_1\sim \mathcal{N}(0,I_M)$, set $x^{(m)}_H=X_{n}^{(m)}$ and run the Hamiltonian flow:\\
\For{$l_f=1,...L-1$}{
$\begin{array}{c} P_H^{(m)}((l_f+\frac{1}{2})\delta)=P_H^{(m)}(l_f)-\frac{\delta}{2}\frac{\partial \Psi}{\partial x_H} (x^{(m)}_H(l_f)) \\ 
x_H^{(m)} ((l_f+1)\delta)=x_H^{(m)}(l_f)+\delta P_H^{(m)}((l_f+\frac{1}{2})\delta)I_M^{-1} \\ 
P_H^{(m)}((l_f+1)\delta))=P_H^{(m)}((l_f+\frac{1}{2})\delta)-\frac{\delta}{2}\frac{\partial \Psi}{\partial x_H}(x^{(m)}_H((l_f+1)\delta))
\end{array}$
}
Calculate acceptance probability and set $x^*=x_H^{(m)}(L)$, $P^*=P_H^{(m)}(L)$:
\begin{equation}
a = 1 \wedge e^{(-\mathcal{H}(x^*,P^*)+\mathcal{H}(x_{H}^{(m)},P_{H}^{(m)}))\Delta t};
\end{equation}
Draw $\mathfrak{u}\sim \mathcal{U}$nif$(0,1)$;\\
\If{$\mathfrak{u}<a$}{
Set $X_{n+1}^{(m)}=x^*$;\\
\Else{Reject, and set $X_{n+1}^{(m)}=x_H^{(m)}$}
}

\If{$X^{m}_n, X^{m}_{n+1}  \in A$}{
Set $X^{m}_n=0, X^{m}_{n+1}=0$
}

}
}
Compute:
\begin{equation*}
\widehat{C}^{HFMC}=\frac{1}{n_S}\sum_{m=1}^{n_S} \left(f(X_{n_t}^{(m)})\prod_{n=1}^{n_t} e^{(-\mathcal{H}(X^{(m)}_{n+1},P^{(m)}_{n+1})+\mathcal{H}(X_{n}^{(m)},P_{n}^{(m)})\Delta t} \mathbb{1}_{X_{n}^{(m)},X_{n+1}^{(m)} \in A_n} \right)
\end{equation*}

    \caption{Hamiltonian Flow Monte Carlo in Rare event setting}
\end{algorithm}

\section{Convergence Analysis}
\subsection{IPS convergence}
IPS convergence, and in particular the asymptotic behaviour as number of particles $n_S\rightarrow \infty$ was thoroughly studied in \cite{DM2004}.

The following result given in \cite{CDG11} allows a non asymptotic control of variance of the rare event probability. 
\begin{assumption} \label{ass4.1}
\begin{equation}
\widetilde{\delta}_n:=\sup_{x,y}\frac{\omega_n(x)}{\omega_n(y)}<+\infty
\end{equation}
\end{assumption}
\begin{theorem}
When the assumption (\ref{ass4.1}) is met for some $\widetilde{\delta}_n$, we have the nonasymptotic estimates:
\begin{equation}
\mathbb{E}\left[\left| \frac{C^{IPC}}{C} -1\right|^2 \right]\leq \frac{4}{n_S}\sum_{s=1}^{n_t}\frac{\widehat{\delta}_s^{(n_t)}}{p_k}
\end{equation}
where $\widehat{\delta}_s^{(n_t)}=\prod_{s\leq k <s+n_t}\widetilde{\delta}_k$
\end{theorem}
\subsection{HFMC Convergence}

\subsubsection{LLN and Convergence Rate}
\hspace{3ex} 
Birkhoff ergodic theorem allows us have law of large numbers(LLN) like convergence. So, we  are interested in a sigma-algebra $\mathcal{G}$ of invariant events, in particular when $\mathcal{G}$ is trivial. 

From lemma 3.3 we know that the Markov chain generated by HFMC is irreducible, and we can see that the Markov Chain that we get from the rare event kernel $\mathcal{M}$ satisfies irreducibility conditions due to the fact that the transition density is always positive. Applying the results by \cite{MT2009}, we have:
\begin{proposition} \cite{MT2009}
Suppose that $\mathcal{M}$ is a $\pi$-irreducible Metropolis kernel. Then $\mathcal{M}$ is a Harris reccurent.
\end{proposition}

\begin{proposition} \cite{MT2009}
If $\mathcal{M}$ is positive Harris and aperiodic then for every initial distribution $\lambda$:
\begin{equation}
||\int \lambda(dx)(\mathcal{M})^l(x,\cdot)-\pi||_{TV}\rightarrow 0, \ \  l\rightarrow \infty
\end{equation}
for $\pi$ almost all $x$.
\end{proposition}
where $||\cdot||_{TV}$ is a total variation distance.



\subsubsection{Geometric Ergodicity}
To establish central limit theorem (CLT), we need a geometric ergodicity of the chain.

\begin{definition}
A subset $C$ of that state space $(\mathbb{R}^{n_S}, \mathcal{B}(\mathbb{R}^{n_S}))$ is petite if there exists a non-zero positive measure $\nu$ on the state space and subsampling distribution $q$  such that
\begin{equation}
\mathcal{K}_q(x,A)\geq \nu(A), \ \forall A \in \mathcal{B}(\mathbb{R}^{n_S}) \ \mbox{and} \ x\in C
\end{equation}
\end{definition}

\begin{definition}
A subset $C$ of that state space $(\mathbb{R}^{n_S}, \mathcal{B}(\mathbb{R}^{n_S}))$ is small if there exists a non-zero positive measure $\nu$ on the state space and real-valued number $l\in \mathbb{R}$ such that
\begin{equation}
\mathcal{K}^l(x,A)\geq \nu(A), \ \forall A \in \mathcal{B}(\mathbb{R}^{n_S}) \ \mbox{and} \ x\in C
\end{equation}
\end{definition}

Observe, that every small set is petite.
\begin{theorem}
Suppose for an irreducible, aperiodic Markov chain having transition probability kernel $\mathcal{K}$ and a state space $\mathbb{R}^{n_S}$, there exists a petite set $C$ areal valued function $V$, satisfying $v\geq 1$, and constants $b<\infty$ and $\lambda<1$ such that
\begin{equation} \label{geom}
\mathcal{K}V(x)\leq \lambda V(x)+b\mathbb{1}_C(x), \forall x \in \mathbb{R}^{n_S} 
\end{equation}
holds. Then the chain is geometrically ergodic.
\end{theorem}
The function $V$ is called a geometric drift. Take the expectation of the both sides of (\ref{geom}) and using the invariance of measure $\pi$ with respect to the kernel $\mathcal{K}$:
\begin{equation}
\mathbb{E}^{\pi}[V(X)]\leq \frac{b\pi(C)}{(1-\lambda)}
\end{equation} 
In other words, for $\lambda \in (0,1]$ a function satisfying (\ref{geom}) is always $\pi$-integrable.
\begin{proposition}
Assume that there exist $\lambda\in [0,1)$ and $b\in\mathbb{R}_+$ such that
\begin{equation}
\mathcal{K} V\leq \lambda V+b
\end{equation}
and 
\begin{equation}
\lim\sup \mathcal{K}(x,\mathcal{R}(x)\cap\mathcal{B}(x))=0
\end{equation}
\end{proposition}

It was shown in \cite{DMS2017} that under certain conditions, HFMC kernel is geometrically ergodic.

\section{Applications and Numerical Results}

\hspace{3ex} We will test our algorithm on  down-out(DOC) Barrier option pricing, and compare its estimate with a standard Monte Carlo and Particle Monte Carlo methods. Lets consider a toy example and assume that our asset follows the following SDE:
\begin{equation} \label{SDE1}
dX_t=\mu X_t dt + \sigma X_tdW_t
\end{equation}
where $\mu$ is a drift, $\sigma$ is a constant volatility parameter.
European DOC call Barrier option is a usual call option contract that pays a payoff $\max(S_T-K,0)$, provided that the asset price $S$ has not fallen below a barrier $B$ during the lifetime of the option. If the pricing process ever reaches the barrier $B$, then the option becomes worthless.

We use Euler-Muruyama disretization scheme and we use the following notation $X_{t_n}:=X_n$, so for a time discretization: $0=t_0,t_1,...,t_{n_t}=T$, the solution of the SDE in (\ref{SDE1}):
\begin{equation}
X_n=X_{n-1}e^{(\mu-0.5\sigma^2)\Delta t+\sigma \Delta t \epsilon_n}
\end{equation}

The DOC barrier call option price of a discretely monitored barrier at maturity $T$ is:
\begin{equation}
C=e^{-(r-q)T}\mathbb{E}[g(X_{n_S})\prod_{n=1}^{n_t} \mathbb{1}_{X_{t\in [t_{n-1},t_{n}]}\in A_n}]
\end{equation}

where $r,q$ are respectively an interest and a dividend rates,  $g(x)=(x-K)^+$ is a payoff function and the set $A_n$ in the case of a DOC barrier call option:
\begin{equation*}
A_n=\inf_{t_{n-1} \leq t \leq t_n}\{t:X_t>B\}
\end{equation*}
We use continuity correction that was proposed in \cite{BGK1997}: $B=B\exp^{-0.5826 \sigma \Delta t}$.

HFMC estimator to compute DOC call option is given by:
\begin{equation}
\widehat{C}^{HFMC}=e^{-(r-q)T} \frac{1}{n_S}\sum_{m=1}^{n_S} \left(g(X_{n_t}^{(m)})\prod_{n=1}^{n_t} e^{(-\mathcal{H}(X^{(m)}_{n+1},P^{(m)}_{n+1})+\mathcal{H}(X_{n}^{(m)},P_{n}^{(m)})\Delta t} \mathbb{1}_{X_{n}^{(m)},X_{n+1}^{(m)} \in A_n} \right)
\end{equation}

Monte Carlo estimate is given by:
\begin{equation}
\widehat{C}^{MC}=e^{-(r-q)T} \frac{1}{n_S}\sum_{m=1}^{n_S} \left( g(X_{n_t}^{(m)})\prod_{n=1}^{n_t} \mathbb{1}_{X_{n}^{(m)}\in A_n}\right)
\end{equation}

The IPS estimator is given by:
\begin{equation*}
\widehat{C}^{IPS}=e^{-(r-q)T} \frac{1}{n_S}\sum_{m=1}^{n_S} \left(g(\widehat{X}_{n_t}^{(m)})\prod_{t=1}^{n_t} W_{n-1}(\textbf{X}_{n-1}^{m}) \mathbb{1}_{\widehat{\textbf{X}}_n \in  A_n} \right) 
\end{equation*}

In the context of a rare event, we chose the barrier level at $65$, with an initial price $X_0=100$, Strike $K=100$, interest rate $r=0.1$, volatility $\sigma=0.3$, $T=0.5$ and zero dividends $q=0$. In the table $1$ and $2$, the Hamiltonian Flow MC, MC and IPS are presented. We used $50000$ and $75000$ particles with $750$ equally spaced time steps in Table $1$ and Table $2$. 

It is very important to choose the number  and the size of leapfrog steps. We chose them such that the acceptance probability $\alpha$ is bigger than $0.8$. 

We compare each approach by estimating the standard deviations, root mean squared error (RMSE), bias, relative mean squared error(RRMSE), time required to compute each estimate and the figure of merit (FOM). We run 20 MC experiments. The RMSE estimator is given by:
\begin{equation}
RMSE = \sqrt{\frac{1}{M_s}\sum_{l=1}^{M_s} ||C-\widehat{C}_{l}||^2}
\end{equation}
where $C$ is price computed analytically, $\widehat{C}_{l}$ are Monte Carlo estimates and $M_s$ is the number of Monte Carlo experiments. 


The RRMSE is computed using the following formula:
\begin{equation}
RRMSE = \frac{RMSE}{\widehat{C}}
\end{equation}

To measure the efficiency of each method presented in the article, we will use the figure of merit(FOM):
\begin{equation}
FOM = \frac{1}{R^2\times CPU_t}
\end{equation}

where $CPU_t$ is CPU time need to compute the estimator and $R$ is a relative error, which is the measure of a statistical precision:
\begin{equation}
R = \frac{St. dev}{\bar{C}} \propto \frac{1}{\sqrt{n_S}}
\end{equation}
where $\bar{C}=\sum_{l=1}^{M_s} \widehat{C}_{l}$ .

\begin{table}[H]
 \caption{DOC Barrier option estimates statistics. $B=65,X_0=100, \ K=100, \ r=0.1, \sigma = 0.3, \ T=1/2$, and $div=0$; $\delta=0.0001$,  \#(Leap frog step): $35$. True price: $10.9064$, $n_S=50000$, $n_t=750$}                                    

\centering                                                                          
 \begin{tabular}{|c|c|c|c|}                                                     
\hline                                                                       
Stat  & MC &   PMC &    HFMC   \\
 
\hline
St. dev. & 0.088518965 & 0.08562686 & 0.065318495\\
RMSE &  0.007011127 & 0.008004332 & 0.0143\\
RRMSE & 0.001298078 & 0.000292621 & 1.87148E-05\\
CPU time & 3.7251 & 4.8432 & 5.90675\\
FOM & 4097.9 & 3387.2 & 4737.6\\
\hline
\end{tabular}
\end{table}

\begin{table}[H]
 \caption{DOC Barrier option estimates statistics. $B=65, X_0=100, \ K=100, \ r=0.1, \sigma = 0.3, \ T=1/2$, and $div=0$; $\delta=0.0009$, \#(Leap frog step): $40$. True price: $10.9064$, $n_S=75000$, $n_t=750$}                                    

\centering                                                                          
 \begin{tabular}{|c|c|c|c|}                                                     
\hline                                                                       
Stat  & MC &   PMC &    HFMC   \\
 
\hline
St. dev. & 0.062385996 & 0.044259477 & 0.038039517\\
RMSE & 0.037561882 & 0.051285344 & 0.037561882\\
RRMSE & 0.000355199 & 0.000240548 & 0.000129293\\
CPU time & 2.2626 & 6.0322 & 7.6832\\
FOM & 13475.2 & 10117.7 & 10711.0 \\
\hline
\end{tabular}
\end{table}

We run $20$ independent Monte Carlo experiments for each estimate. Since IPS and the simulation of the Hamiltonian dynamics requires more time to compute an estimate, we use the figure of merit to compare three approaches. From the table 1 and 2 we can observe that HFMC demonstrates standard deviations, bias and relative RMSE. 

\section{Conclusion and Further Research} 

\hspace{3ex} We proposed an importance sampling algorithm based on the simulation of the Hamiltonian system, that generates a Markov Chain that follows along the gradient of the target distributions over large distances of the state space, while producing low-variance samples. 

From the simulated results we saw that HFMC allows efficiently estimate rare event probabilities, which we tested on the case of DOC Barrier options. Its estimates show lower variance and bias than that of MC and IPS.

It will interesting to adapt a stochastic gradient Hamiltonian Monte Carlo algorithm \cite{FG2014}, when one can avoid computing the gradient at each simulations. Taking into account the big data problem and the necessity of online estimations, we can get sufficient improvements. Another extension is the adaptation to the Riemann Manifold Hamiltonian Monte Carlo \cite{GC2011}, when we can create a statistical manifold and tune HFMC by computing explicitely the mass matrix $M$ in the kinetic energy of the algorithm.
 
In the next article we will show the performance of mixed IPS and the Hamiltonian Flow Monte Carlo. It will allow faster explore the state space on the one hand, and push trajectories into rare event area on the other hand. By resampling we can reduce the correlation between generated from the Hamiltonian system Markov chains.




\end{document}